\documentclass{llncs}

\usepackage[usenames]{color}
\usepackage{graphicx}
\usepackage{amsmath}
\usepackage{amssymb}
\usepackage{envmath}
\usepackage{amsfonts}
\usepackage{fullpage}
\usepackage{framed}
\usepackage{xkvltxp}
\usepackage{subfigure}
\usepackage[backref=true,colorlinks=true]{hyperref}

\newcommand{\eps}{\ensuremath{\varepsilon}}
\newcommand{\dspace}{\mathcal{S}}
\newcommand{\runs}{\ensuremath{\operatorname{runs}}}
\newcommand{\bwt}{\ensuremath{\operatorname{bwt}}}

\newcommand{\polylog}{\ensuremath{\operatorname{polylog}}}
\newcommand{\patrascu}{P\v{a}tra\c{s}cu}
\newcommand{\poly}{\ensuremath{\operatorname{poly}}}

\newcommand{\strlen}{L}
\newcommand{\blocksize}{B}
\newcommand{\blockcount}{N}
\newcommand{\instcount}{H}

\newcommand{\seta}{{\mathbf X}}
\newcommand{\setb}{{\mathbf Y}}

\newcommand{\setsize}{\blockcount}

\let\doendproof\endproof
\renewcommand\endproof{~\hfill\qed\doendproof}

\title{Data Structure Lower Bounds on Random Access to Grammar-Compressed Strings}
\author{Shiteng Chen\inst{1} \and Elad Verbin \inst{2} \and Wei Yu\inst{2}}

\institute{Tsinghua University, China \\ \email{shitengchen@gmail.com} \and Aarhus University\thanks{The authors acknowledge support from the Danish National Research Foundation and The National Science Foundation of China (under the grant
61061130540) for the Sino-Danish Center for the Theory of Interactive Computation, within which part of this work was performed.}, Denmark \\ \email{eladv@cs.au.dk,yuwei@cs.au.dk}}

\begin{document}

\maketitle

\begin{abstract}
In this paper we investigate the problem of building a static data structure that represents a string $s$ using space close to its compressed size, and allows fast access to individual characters of $s$. This type of structures was investigated by the recent paper of Bille et al.~~\cite{bille-2010}. Let $n$ be the size of a context-free grammar that derives a unique string $s$ of length $\strlen$. (Note that $\strlen$ might be exponential in $n$.) Bille et al.\ showed a data structure that uses space $O(n)$ and allows to query for the $i$-th character of $s$ using running time $O(\log \strlen)$. Their data structure works on a word RAM with a word size of $\log \strlen$ bits.

Here we prove that for such data structures, if the space is $\poly(n)$, then the query time must be at least $(\log \strlen)^{1-\eps}/\log \dspace$ where $\dspace$ is the space used, for any constant $\eps>0$. As a function of $n$, our lower bound is $\Omega(n^{1/2-\eps})$. Our proof holds in the cell-probe model with a word size of $\log \strlen$ bits, so in particular it holds in the word RAM model. We show that no lower bound significantly better than $n^{1/2-\eps}$ can be achieved in the cell-probe model, since there is a data structure in the cell-probe model that uses $O(n)$ space and achieves $O(\sqrt{n \log n})$ query time. The ``bad'' setting of parameters occurs roughly when $\strlen=2^{\sqrt{n}}$. We also prove a lower bound for the case of not-as-compressible strings, where, say, $\strlen=n^{1+\eps}$. For this case, we prove that if the space is $n \cdot polylog(n)$, then the query time must be at least $\Omega(\log n / \log \log n)$.

The proof works by reduction to communication complexity, namely to the LSD  (Lopsided Set Disjointness) problem, recently employed by \patrascu\ and others.
We prove lower bounds also for the case of LZ-compression and Burrows-Wheeler (BWT) compression. All of our lower bounds hold even when the strings are over an alphabet of size $2$ and hold even for randomized data structures with 2-sided error.
\end{abstract}

\section{Introduction}

In many modern databases, strings are stored in compressed form. Many compression schemes are grammar-based, in particular Lempel-Ziv~\cite{lempel1976complexity,ziv1977universal,ziv1978compression} and its variants, as well as
Run-Length Encoding. Another big family of textual compressors are BWT (Burrow-Wheeler Transformation~\cite{burrows1994block}) based compressors, like the one used by the software bzip2.

A natural desire is to store a text using space close to its compressed size, but to still allow fast access to individual characters: can we do something faster than simply extracting the whole text each time we need to access a character? This question was recently answered in the affirmative by Bille et al.~\cite{bille-2010} and by Claude and Navarro~\cite{claude2009self}. These two works investigate the problem of storing a string that can be represented by a small CFG (context-free grammar) of size $n$, while allowing some basic stringology operations, in particular random access to a character in the text. The data structure of Bille et al.~\cite[Theorem~1]{bille-2010} stores the text in space linear in $n$, while allowing access to an individual character in time $O(\log \strlen)$, where $\strlen$ is the text's \emph{uncompressed} size. (The result of Bille et. al.\ also allows other query operations such as pattern matching; we do not discuss those in this paper.) But is that the best upper bound possible?

In this paper we show a $(\log \strlen)^{1-\eps}$ lower bound on the query time whenever the space used by the data structure is $poly(n)$, showing that the result of Bille et al.\ is close to optimal. Our lower bounds are proved in the cell-probe model of Yao~\cite{yao81}, with word size $\log \strlen$, therefore they in particular hold for the model studied by Bille et al.~\cite{bille-2010}, since the cell-probe model is strictly stronger than the RAM model. Our lower bound is proved by a reduction to Lopsided Set Disjointness (LSD), a problem for which \patrascu\ has recently proved an essentially-tight randomized lower bound~\cite{patrascu8unifying}. The idea is to prove that grammars are rich enough to effectively ``simulate'' a disjointness query: our class of grammars, presented in Section~\ref{sec:reduction_from_lsd}, might be of independent interest as a class of ``hard'' grammars for other purposes as well.

In terms of $n$, our lower bound is $n^{1/2-\eps}$. The results of Bille et al.\ imply an upper bound of $O(n)$ on the query time, since $\log \strlen \le n$, therefore in terms of $n$ there is a curious quadratic gap between our lower bound and Bille et al.'s upper bound. We show that this gap can be closed by giving a better data structure: we show a data structure which takes space $O(n)$ and has query time $O(\sqrt{n \log n})$, showing that no significantly better lower bound is possible. This data structure, however, comes with a big caveat -- it runs in the highly-unrealistic cell-probe model, thus serving more as an impossibility proof for lower bounds than as a reasonable upper bound. The question remains open of whether such a data structure exists in the more realistic word RAM model.

Our lower bound holds for a particular, ``worst-case'', dependence of $\strlen$ on $n$. Namely, $\strlen$ is roughly $2^{\sqrt{n}}$. It might also be interesting to explicitly limit the range of allowed parameters to other regimes, for example to non-highly-compressible text; in such a regime it might be that $\strlen =n^{1+\eps}$. The above result does not imply any lower bound for this case. Thus, we show in another result that for any data structure in that regime, if the space is $n \cdot \polylog n$, then the query time must be $\Omega(\log n/\log \log n)$. This lower bound holds, again, in the cell probe model with words of size $\log n$ bits, and is proved by a reduction to two-dimensional range counting (which, once again, was lower bounded by a reduction to LSD~\cite{patrascu8unifying}). To this end, we prove a new lower bound for two-dimensional range counting on $[n]\times [n^\eps]$ grid, which was not previous known, and is of independent interest.

\section{Preliminaries} \label{sec:prelim}

In this paper we denote $[m]=\{1,\ldots,m\}$. All logarithms are in base $2$ unless explicitly stated otherwise.

Our lower bounds are proved in Yao's cell-probe model~\cite{yao81}. In the cell-probe model, the memory is seen as an array of cells, where each cell consists of $w$ bits each. The query time is measured as the number of cells read from memory, and all computations are free. This model is strictly stronger than the word RAM, since in the word RAM the operations allowed on words are restricted, while in the cell-probe model we only measure the number of cells accessed. The cell-probe model is widely used in proving data structure lower bounds, especially by reduction to communication complexity problems~\cite{miltersen1995data}. In this paper we prove our result by a reduction to the \textsc{Blocked-LSD} problem introduced by \patrascu~\cite{patrascu8unifying}.

An SLP (straight line program) is a collection of $n$ derivation rules, defining the symbols $g_1,\ldots,g_n$. Each rule is either of the form $g_i \rightarrow `\sigma'$, i.e. $g_i$ is a terminal, which takes the value of a character $\sigma$ from the underlying alphabet, or of the form $g_i \rightarrow g_j g_k$, where $j<i$ and $k<i$, i.e. $g_j$ and $g_k$ were already defined, and we define the nonterminal symbol $g_i$ to be their concatenation. The symbol $g_n$ is the \emph{start symbol}. To derive the string we start from $g_n$ and follow the derivation rules until we get a sequence of characters from the alphabet. The length of the derived string is at most $2^n$. W.l.o.g. we assume it is at least $n$. As the same in Bille et al.~\cite{bille-2010}, we also assume w.l.o.g. that the grammars are in fact SLPs and so on the righthand side of each grammar rule there are either exactly two variables or one terminal symbol. In this paper SLP, CFG and grammar all mean the same thing.

The grammar random access problem is the following problem.
\begin{definition}[Grammar Random Access Problem]
For a CFG $\mathrm G$ of size $n$ representing a binary string of length $\strlen$, the problem is to build a data structure to support the following query: given $1 \le i \le \strlen$, return the $i$-th character (bit) in the string.
\end{definition}

We study two other data-structured problems, which are closely related to their communication-complexity counterparts.
\begin{definition}[Set Disjointness, $SD_\setsize$]
For a set $\setb \subseteq [\setsize]$, the problem is to build a data structure to support the following query: given a set $\seta \subseteq [\setsize]$, answer whether $\seta \cap \setb = \emptyset$.
\end{definition}

Given a universe $[\blocksize \blockcount]$, a set $\seta$ is called \emph{blocked with cardinality $\blocksize$} if when we divide the universe $[\blocksize \blockcount]$ into $\blockcount$ equal-sized consecutive blocks, $\seta$ contains exactly one element from each of the blocks while $\setb$ could be arbitrary.

\begin{definition}[Blocked Lopsided Set Disjointness, $BLSD_{\blocksize,\blockcount}$]
For a set $\setb \subseteq [\blocksize \blockcount]$, the problem is to build a data structure to support the following query: given a blocked set $\seta \subseteq [\blocksize \blockcount]$ with cardinality $\blockcount$, answer whether $\seta \cap \setb = \emptyset$.
\end{definition}

For proving lower bound for near-linear space data structures, we also need a variant of the range counting problem.

\begin{definition}[Range Counting on $\[n\]^{1+\eps}$ Grid] \label{def:rangecounting}
The range counting problem is a static data structure problem. We need to preprocess a set of $n$ points on a $[n] \times [n^\eps]$ grid. A query $(x,y)$ asks to count the number of points in a dominance rectangle $[1,x]\times [1,y]$, \emph{modulo $2$}.
\end{definition}

When $\eps=1$, the above problem has been investigated under the name ``range counting'' in \patrascu~\cite{patrascu8unifying}. Note that the above problem is ``easier'' than the classical 2D range-counting problem, since it is a dominance problem, it is on the grid, and it is modulo 2. However, the (tight) lower bound that is known for the general problem in \patrascu~\cite{patrascu8unifying}, is proven for the problem we define. In this paper, we extend this result a little bit to give a lower bound on the universe of $[n] \times [n^\eps]$ for any constant $\eps$.

\section{Lower Bound for Grammar Random Access}

In this section we prove the main lower bound for grammar random access. In Section \ref{sec:reduction_from_lsd} we show the main reduction from SD and BLSD. In Section \ref{sec:sd_and_lsd_lowerbounds} we prove lower bounds for SD and BLSD, based on reductions to communication complexity (these are implicit in the work of \patrascu~\cite{patrascu8unifying}). Finally, in Section \ref{sec:tying_it_together} we tie these together to get our lower bounds.

\subsection{Reduction from SD and LSD} \label{sec:reduction_from_lsd}

In this section we show how to reduce the grammar access problem to SD or BLSD, by considering a particular type of grammar. The reductions tie the parameters $n$ and $\strlen$ to the parameters $\blocksize$ and $\blockcount$ of BLSD (or just to the parameter $\setsize$ of SD). In Section \ref{sec:tying_it_together} we show how to choose the relation between the various parameters in order to get our lower bounds. We remark that the particular multiplicative constants in the lemmas below will not matter, but we give them nonetheless, for concreteness.

These reductions might be confusing for the reader, but they are in fact almost entirely tautological. They just follow from the fact that the communication matrix of SD is a tensor product of the 2 by 2 communication matrices for the coordinates, i.e., it is just a $\blockcount$-fold tensor product of the matrix $\left(
\begin{array}{cc}
1 & 1 \\
1 & 0 \\
\end{array}
\right)$.
For BLSD, the communication matrix is the $\blockcount$-fold tensor product of the $(2^\blocksize) \times \blocksize$ communication matrix for each block (for example, for $\blocksize=3$ this matrix is
$\left(
  \begin{array}{cccccccc}
    1 & 0 & 1 & 0 & 1 & 0 & 1 & 0 \\
    1 & 1 & 0 & 0 & 1 & 1 & 0 & 0 \\
    1 & 1 & 1 & 1 & 0 & 0 & 0 & 0 \\
  \end{array}
\right)$). We do not formulate our arguments in the language of communication matrices and tensor products, since this would hide what is really going on. To aid the reader, we give an example after each of the two constructions.

\begin{lemma}[Reduction from $SD_{\setsize}$] \label{lem:reduction_from_sd}
For any set $\setb \subseteq [\setsize]$, there is a grammar $G_{\setb}$ of size $n=2\setsize+1$ deriving a binary string $s_{\setb}$ of length $\strlen=2^{\setsize}$ such that for any set $\seta \subseteq [\setsize]$, it holds that $s_{\setb}[\seta]=1$ iff $\seta \cap \setb = \emptyset$.
\end{lemma}

Note that in this lemma we have indexed the string $s$ by \emph{sets}: there are $2^{\setsize}$ possible sets $\seta$, and the length of the string $s_{\setb}$ is also $2^{\setsize}$ -- each set $\seta$ serves as an index of a unique character. The indexing is done in lexicographic order: the set $\seta$ is identified with its \emph{characteristic vector}, i.e., the vector in $\{0,1\}^{\setsize}$ whose $i$-th coordinate is `1' if $i \in \seta$, and `0' otherwise, and the sets are ordered according to lexicographic order of their characteristic vectors. For example, here is the ordering for the case $\setsize=3$: $\emptyset, \{1\},\{2\},\{1,2\}, \{3\},\{1,3\},\{2,3\},\{1,2,3\}$.

\begin{proof}
We now show how to build the grammar $G_{\setb}$. The grammar has $\setsize$ symbols for the strings $0,0^2,0^4,\ldots,0^{2^{\setsize-1}}$, i.e., all strings consisting solely of the character `0', of lengths which are all powers of 2 up to $2^{\setsize-1}$. Then, the grammar has $\setsize+1$ additional symbols $g_0,g_1,\ldots,g_\setsize$. The terminal $g_0$ is equal to the character $1$. For any $1 \le i \le \setsize$, we set $g_i$ to be equal to $g_{i-1}g_{i-1}$ if $i \notin \setb$, and to be equal to $g_{i-1}0^{2^{i-1}}$ if $i \in \setb$. The start symbol of the grammar is $g_{\setsize}$.

We claim that the string derived by this grammar has the property that $s_{\setb}[\seta]=1$ iff $\seta \cap \setb = \emptyset$. This is easy to prove by induction on $i$, where the induction claim is that for any $i$, $g_i$ is the string that corresponds to the set $\setb \cap \{1,\ldots,i\}$ over the universe $\{1,\ldots,i\}$.
\end{proof}

\begin{example} \label{exm:sdexample}
Consider the universe $\setsize=4$. Let $\setb=\{1,3\}$. The string $s_{\setb}$ is $1010000010100000$. The locations of the 1`s correspond exactly to the sets that don't intersect $\setb$, namely to the sets $\emptyset$, $\{2\}$, $\{4\}$ and $\{2,4\}$, respectively.
\end{example}

We now show the reduction from blocked LSD. It follows along the same general idea, but the grammar is slightly more complicated.

\begin{lemma}[Reduction from $BLSD_{\blocksize,\blockcount}$] \label{lem:reduction_from_blsd}
For any set $\setb \subseteq [\blocksize \blockcount]$, there is a grammar $G_{\setb}$ of size $n=2\blocksize \blockcount+1$ deriving a binary string $s_{\setb}$ of length $\strlen=\blocksize^\blockcount$ such that for any blocked set $\seta \subseteq [\blocksize \blockcount]$ of cardinality $m$, it holds that $s_{\setb}[\seta]=1$ iff $\seta \cap \setb = \emptyset$.
\end{lemma}

Recall that by a ``blocked set $\seta \subseteq [\blocksize \blockcount]$ of cardinality $\blockcount$'' we mean a set such that the universe $[\blocksize \blockcount]$ is divided into $\blockcount$ equal-sized blocks, and $\seta$ contains exactly one element from each of these blocks.

Note that in this lemma we have again indexed the string $s$ by \emph{sets}: there are $\blocksize^\blockcount$ possible sets $\seta$ and the length of the string is $\blocksize^\blockcount$. The indexing is done in \emph{lexicographic order}, this time identifying a set $\seta$ with a length-$\blockcount$ vector whose $i$-th coordinate is chosen according to which element it contains in block $i$, and the sets are ordered according to lexicographic order of their characteristic vectors. For example, here is the ordering for the case $\blockcount=2,\blocksize=3$: $\{1,4\},\{2,4\},\{3,4\},\{1,5\},\{2,5\},\{3,5\},\{1,6\},\{2,6\},\{3,6\}$.

The construction in this reduction is similar to that in the case of $SD$, but instead of working element by element, we work block by block.

\begin{proof}
We now show how to build the grammar $G_{\setb}$. The grammar has $\blockcount$ symbols for the strings $0,0^\blocksize,0^{\blocksize^2},0^{\blocksize^3},\ldots,0^{\blocksize^{\blockcount-1}}$, i.e., all strings consisting solely of the character `0', of lengths which are all powers of $\blocksize$ up to $\blocksize^{\blockcount-1}$. We cannot simply obtain the symbols directly from each other: e.g., to obtain \ $0^{\blocksize^2}$ from $0^\blocksize$, we need to concatenate $0^\blocksize$ with itself $\blocksize$ times. Thus we use $\blocksize \blockcount$ rules to derive all of these symbols. (In fact, $O(\blockcount \log \blocksize)$ rules can suffice but this does not matter).

Then, beyond these, the grammar has $\blockcount+1$ additional symbols $g_0,g_1,\ldots,g_\blockcount$, one for each block. The terminal $g_0$ is equal to the character $1$. For any $1 \le 1 \le \blockcount$, $g_i$ is constructed from $g_{i-1}$ according to which elements of the $i$-th block are in $\setb$: we set $g_i$ to be a concatenation of $\blocksize$ symbols, each of which is either $g_{i-1}$ or $0^{\blocksize^{i-1}}$. In particular, $g_i$ is the concatenation of $g_i^{(1)},\ldots,g_i^{(\blocksize)}$, where $g_i^j$ is equal to $g_{i-1}$ if the $j$-th element of the $i$-th block is not in $\setb$, and it is equal to $0^{\blocksize^{i-1}}$ if the $j$-th element of the $i$-th block is in $\setb$. To construct these symbols we need at most $\blocksize \blockcount$ rules, because we need $\blocksize-1$ concatenation operations to derive $g_i$ from $g_{i-1}$. (Note that here we cannot get down to $O(\blockcount \log \blocksize)$ rules -- $\Theta(\blocksize \blockcount)$ seem to be necessary.) The start symbol of the grammar is $g_\blockcount$.

We claim that the string produced by this grammar has the property that $s_{\setb}[\seta]=1$ iff $\seta \cap \setb = \emptyset$. This is easy to prove by induction on $i$, where the induction claim is that for any $i$, $g_i$ is the string that corresponds to the set $\seta \cap \{1\,\ldots,i \blocksize\}$ over the universe $\{1,\ldots,i \blocksize \}$.
\end{proof}

\begin{example} \label{exm:blsdexample}
Consider the values $\blocksize=3$ and $\blockcount=3$. Let $\setb=\{1,3,5,9\}$. The string $s_{\setb}$ is $010000010010000010000000000$. The locations of the 1's correspond exactly to the blocked sets that don't intersect $\setb$, namely to the sets $\{2,4,7\}$, $\{2,6,7\}$, $\{2,4,8\}$ and $\{2,6,8\}$, respectively. A brief illustration for this example is in Figure~\ref{fig:blsdexample}.
\end{example}
\begin{figure}[!ht]
  \centering
  \includegraphics[width=8cm]{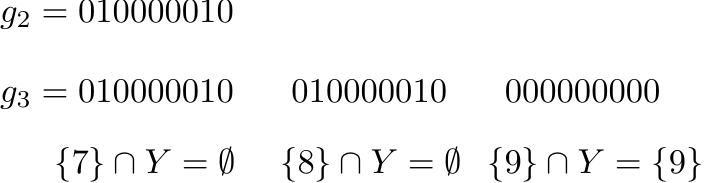}
  \caption{An illustration of Example~\ref{exm:blsdexample}.}
 \label{fig:blsdexample}
\end{figure}
\subsection{Lower bounds for SD and BLSD} \label{sec:sd_and_lsd_lowerbounds}
In this subsection we show lower bounds for SD and BLSD that are implicit in the work of \patrascu~\cite{patrascu8unifying}. Recall the notations from Section \ref{sec:prelim}: in particular, in all of the bounds, $w$, $\dspace$, and $t$ denote the word size (measured in bits), the size of the data structure (measured in words) and the query time (measured in number of accesses to words), respectively.

\begin{theorem} \label{thm:sd_ds_lb}
For any 2-sided-error data structure for $SD_\setsize$, $t \ge \Omega(\setsize / ( w + \log S))$.
\end{theorem}

Note that this theorem does not give strong bounds when $w = O(\log \strlen)$, but it is meaningful for bit-probe ($w=1$) bound and a warm-up for the reader.

\begin{theorem} \label{thm:blsd_ds_lb}
Let $\eps>0$ be any small constant. For any 2-sided-error data structure for $BLSD_{\blocksize,\blockcount}$,
\begin{equation} \label{eq:blsd_ds_lb}
t \ge \Omega \left( \min \left( \frac{\blockcount \log \blocksize}{\log \dspace}, \frac{\blocksize^{1-\eps}\blockcount}{w} \right) \right) \ .
\end{equation}
\end{theorem}

The proofs follow by standard reductions from data structure to communication complexity, using known lower bounds for SD and BLSD (the latter is one of the main results in \cite{patrascu8unifying}).

We now cite the corresponding communication complexity lower bounds:

\begin{lemma}[See~\cite{bar2004information,razborov1992distributional,babai1986complexity}] \label{thm:sd_comm_lb}
Consider the communication problem where Alice and Bob each receive a subset of $m$, and they want to decide whether the sets are disjoint. Any randomized 2-sided-error protocol for this problem uses communication $\Omega(\setsize)$.
\end{lemma}

\begin{lemma}[See~\cite{patrascu8unifying}, Lemma 3.1] \label{thm:blsd_comm_lb}
Let $\eps>0$ be any small constant. Consider the communication problem where Bob gets a subset of $[\blocksize \blockcount]$ and Alice gets a blocked subset of $[\blocksize\blockcount]$ of cardinality $\blockcount$, and they want to decide whether the sets are disjoint. In any randomized 2-sided-error protocol for this problem, either Alice sends $\Omega(\blockcount \log \blocksize)$ bits or Bob sends $\blocksize^{1-\eps} \blockcount$ bits. (The $\Omega$-notation hides a multiplicative constant that depends on $\eps$).
\end{lemma}

The way to prove the data structure lower bounds from the communication lower bounds is by reductions to communication complexity: Alice and Bob execute a data structure query; Alice simulates the querier, and Bob simulates the data structure. Alice notifies Bob which cell she would like to access; Bob returns that cell, and they continue for $t$ rounds, which correspond to the $t$ probes. At the end of this process, Alice knows the answer to the query. Overall, Alice sends $t \log \dspace$ bits and Bob sends $tw$ bits. The rest is calculations, which we include here for completeness:

\begin{proof}[Lemma~\ref{thm:sd_comm_lb} $\Rightarrow$ Theorem~\ref{thm:sd_ds_lb}]
We know that the players must send a total of $\Omega(\setsize)$ bits, but the data structure implies a protocol where $t \log \dspace + tw$ bits are communicated. Therefore $t \log \dspace + tw \ge \Omega(\setsize)$ so $t \ge \Omega(\setsize / (\log \dspace + w))$.
\end{proof}

\begin{proof}[Lemma~\ref{thm:blsd_comm_lb} $\Rightarrow$ Theorem~\ref{thm:blsd_ds_lb}]
We know that either Alice sends $\Omega(\blockcount \log \blocksize)$ bits or Bob sends $\blocksize^{1-\eps} \blockcount$ bits. Therefore, either $t \log \dspace \ge \Omega(\blockcount \log \blocksize)$ or $tw \ge \blocksize^{1-\eps} \blockcount$. The conclusion follows easily.
\end{proof}

\subsection{Putting it Together} \label{sec:tying_it_together}

We now put the results of Section~\ref{sec:reduction_from_lsd} and \ref{sec:sd_and_lsd_lowerbounds} together to get our lower bounds. Note that in all lower bounds below we freely set the relation of $n$ and $\strlen$ in any way that gives the best lower bounds. Therefore, if one is interested in only a specific relation of $n$ and $\strlen$ (say $\strlen=n^{10}$) the lower bounds below are not guaranteed to hold. The typical ``worst'' dependence in our lower bounds (at least for the case where $w=\log \strlen$ and $\dspace=\poly(n)$) is roughly $\strlen=2^{\sqrt{n}}$.

Theorem \ref{thm:sd_ds_lb} together with Lemma \ref{lem:reduction_from_sd} immediately give:
\begin{theorem} \label{thm:grammar_lb_from_sd}
For any 2-sided-error data structure for the grammar random access problem, $t \ge \Omega(n / ( w + \log \dspace))$. And in terms of $\strlen$, $t \ge \Omega(\log \strlen / ( w + \log \dspace))$.

When setting $w=1$ and $\dspace=poly(n)$ (polynomial space in the bit-probe model), we get that $t \ge \Omega(n/ \log n)$. And in terms of $\strlen$, $t \ge \Omega(\log \strlen / \log \log \strlen)$.
\end{theorem}

\begin{proof}
Trivial, since $n=\Theta(\setsize)$ and $\strlen=2^{\Theta(\setsize)}$.
\end{proof}

Theorem \ref{thm:blsd_ds_lb} together with Lemma \ref{lem:reduction_from_blsd} give:
\begin{theorem} \label{thm:grammar_lb_from_blsd_det}
Assume $w = \omega(\log \dspace)$. Let $\eps>0$ be any arbitrarily small constant. For any 2-sided-error data structure for the grammar random access problem, $t \ge n/w^{\frac{1+\eps}{1-\eps}}$. And in terms of $\strlen$, $t \ge \frac{\log \strlen}{\log \dspace \cdot w^{\frac{\eps}{1-\eps}}}$.

When setting $w=\log \strlen$ and $\dspace =\poly(n)$ (polynomial space in the cell-probe model with cells of size $\log \strlen$), there is another constant $\delta$ such that we get that $t \ge n^{1/2 - \delta}$. And in terms of $\strlen$, $t \ge (\log \strlen)^{1-\delta}$.
\end{theorem}

The condition $w = \omega(\log \dspace)$ is a technical condition, which ensures that the value of $\blocksize$ we choose in the proof is at least $\omega(1)$. For $w \le \log \dspace$ one gets the best results just by reducing from SD, as in Theorem \ref{thm:grammar_lb_from_sd}.

\begin{proof}
For the first part of the theorem, substitute $\blocksize=(w/\log \dspace)^{1/(1-\eps)} \log(w/ \dspace)$, $\blockcount =n/\blocksize$, $\strlen=\blocksize^\blockcount$ into \eqref{eq:blsd_ds_lb}. For the second part of the theorem, substitute $\blockcount =\frac{\blocksize^{1-\eps} \log n}{\log^2 \blocksize}$, $n=\blocksize \blockcount$ and $\strlen=\blocksize^\blockcount$. And for the result, set $\delta=\frac{2\eps}{1-\eps}$.
\end{proof}

\section{Lower Bound for Less-Compressible Strings} \label{sec:rclowerbound}

\subsection{The Range Counting Lower Bound}

In the above reduction, the worst case came from strings that can be compressed superpolynomially. However, for many of the kinds of strings we expect to encounter in practice, superpolynomial compression is unrealistic. A more realistic range is polynomial compression or less. In this section we discuss the special case of strings of length $\Theta(n^{1+\eps})$.
We show that for this class of strings, the Bille et al.~\cite{bille-2010} result is also (almost) tight by proving an $\Omega(\log n/\log \log n)$ lower bound on the query time, when the space used is $O(n \cdot \polylog n)$. This is done by reduction from the range counting problem on a 2D (two-dimensional) grid. Leaving the proof to Appendix~\ref{sec:rclowerbound}, we have the following lower bound for the range counting problem (see Definition~\ref{def:rangecounting} for details).

\begin{lemma} \label{lem:2dlower}
Any data structure for the 2D range counting problem using $O(n \polylog n)$ space requires query time $\Omega(\log n/\log \log n)$ in the cell probe model with cell size $O(\log n)$.
\end{lemma}

Recall that the version of range counting we consider is actually dominance counting modulo 2 on the $n \times n^{\eps}$ grid.

The main idea behind our reduction is to consider the length-$n^{1+\eps}$ binary string consisting of the answers to all $n^{1+\eps}$ possible range queries (in the natural order, i.e. row-by-row, and in each row from left to right); call this the \emph{answer string} of the corresponding range counting instance. We prove that this string can be represented using a grammar of size $O(n \log n)$. The reduction obviously follows, since a range query can then just be answered by asking for one bit of the compressed string.

\begin{lemma} \label{lem:rcreduction}
For any range counting problem in 2D, the answer string can be represented by a CFG of size $O(n \log n)$.
\end{lemma}

The idea behind the proof of the lemma is to simulate a sweep of the point-set from top to bottom by a dynamic one-dimensional range tree. The symbols of the CFG will  correspond to the nodes of the tree. With each new point encountered, only $2 \log n$ new symbols have to be introduced.

\begin{proof}
Assume for simplicity that $n$ is a power of $2$.
It is easy to see that the answer string could be built by concatenating the answers in a row-wise order, just as illustrated in Figure~\ref{fig:answerstring}.

\begin{figure}[!ht]
  \centering
  \includegraphics[width=5cm,height=5cm]{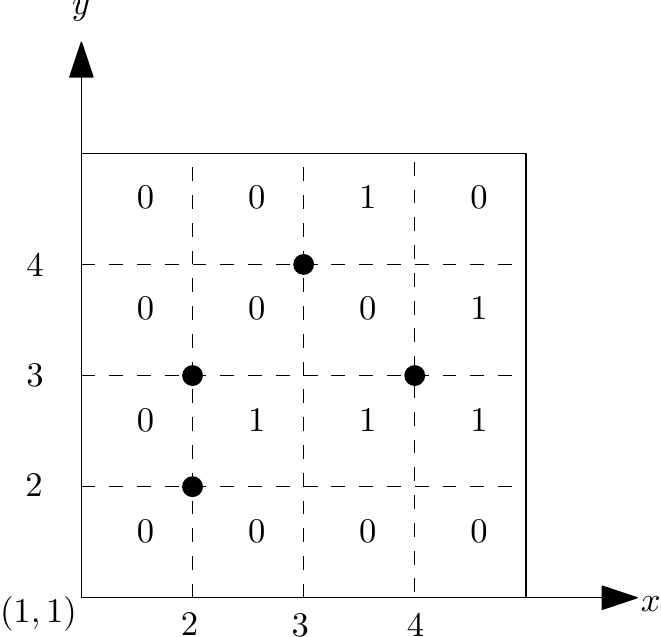}
  \caption{The answer string for this instance is 0000011100010010. The value in the grids are the query results for queries falling in the corresponding cell, including the bottom and left boundaries, excluding the right and top boundaries.}
  \label{fig:answerstring}
\end{figure}

We are going to build the string row by row. Think of a binary tree representing the CFG built for the first row of the input. The root of the tree derives the first row of the answer string, whose two children respectively represent the answer string for the left and the right half of the row. In this way the tree is built recursively. The leaves of the tree are terminal symbols in $\{0,1\}$. Thus there are $2n-1$ symbols in total for the whole tree. At the same time we also maintain the \emph{negations} of the symbols in the tree, i.e., making a new symbol $g'_i$ for each $g_i$ in the tree, where $g'_i = 1-g_i$ if $g_i$ is a terminal symbol, or $g'_i = g'_j g'_k$ if $g_i = g_j g_k$.

The next row in the answer string will be built by changing at most $2p\log n$ symbols in the old tree, where $p$ is the number of new points in the next row. The symbols for the new row are  built by re-using most of the symbols in the old row, and introducing new symbols where needed. We process the new points one by one, and for each one, the modifications needed all lie in a path from a leaf to the root of the tree. Assuming the update is the path $h_{1}, h_{2}, \ldots, h_{\log n}$, the new tree will contain an update of $h_1, h_2, \ldots, h_{\log n}$. Also, all the right children of these nodes will be switched with their negations (this switching step does not actually require introducing any new symbols).
An intuitive picture of the process is given in Figure~\ref{fig:tree}.

\begin{figure}[!ht]
\subfigure[The trees built for the first and second rows for the example in Firgure~\ref{fig:answerstring}.]{
\centering
\includegraphics[width=9cm]{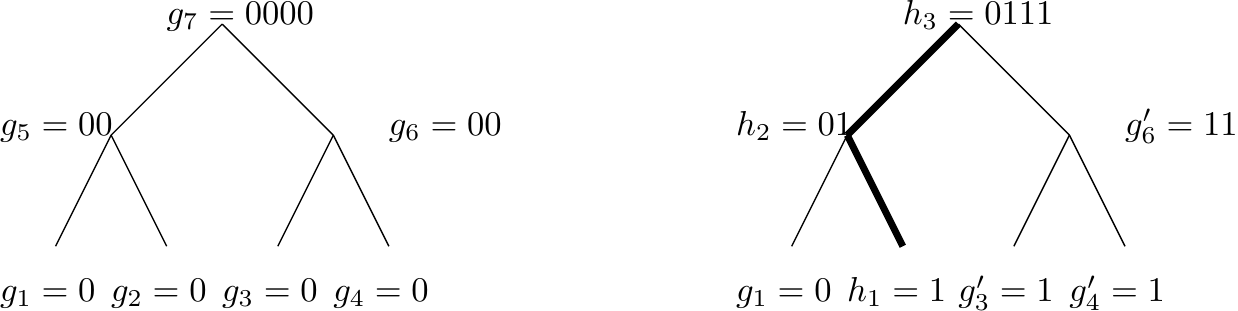}
\label{fig:treea}
}
\hspace{0.5cm}
\subfigure[The general process illustrated by graph. The black parts stands for the negations of corresponding symbols.]{
\centering
\includegraphics[width=5cm]{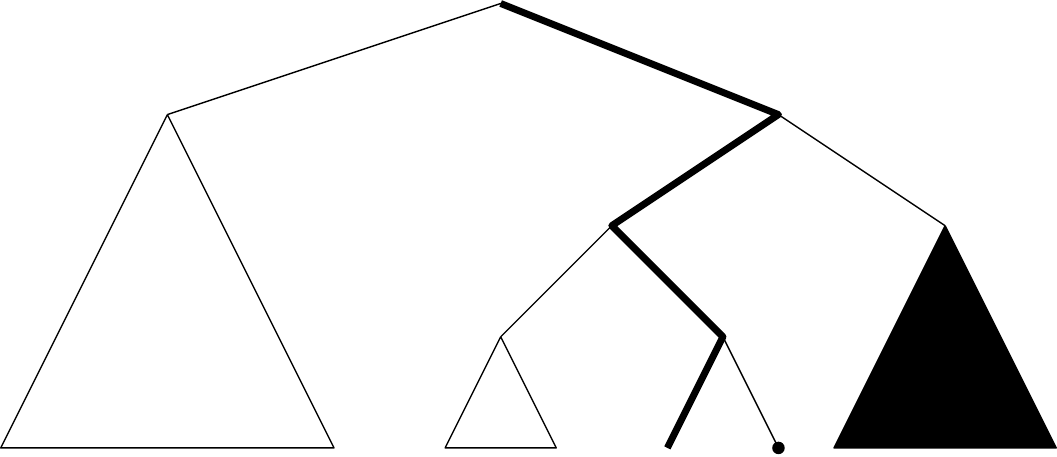}
\label{fig:treeb}
}
\caption{Examples for building answer strings.}
\label{fig:tree}
\end{figure}

It is easy to see for each new point, $2\log n$ additional symbols are created. $\log n$ of them are the new symbols ($h_1, \ldots, h_{\log n}$), and another $\log n$ of them are their negations ($h'_1, \ldots, h'_{\log n}$). After all, we use $2n-1+(n^{\eps}-1)\cdot 2\log n = O(n+n^{\eps}\log n)$ symbols to derive the whole answer string.
\end{proof}

By using the above lemma, we have the lower bound of the grammar random access problem.

\begin{theorem} \label{thm:n2lower}
Any data structure using space $O(n \polylog n)$ for the grammar random access problem requires $\Omega(\log n / \log \log n)$ query time.
\end{theorem}
\begin{proof}
For inputs of the range counting problem, we compress the answer string to a CFG of size $O(n \log n)$ according to Lemma~\ref{lem:rcreduction}. After that we build a data structure for the random access problem on this CFG using Lemma~\ref{lem:up1}. For any query $(x,y)$ of the range counting problem, we simply pass the query result on the index $(y-1) n+x-1$ on the answer string as an answer. By Bille et al.~\cite{bille-2010} this makes a data structure using $O(n \log n)$ space with $t(n\log n)$ query time for the range counting problem. According to Lemma~\ref{lem:2dlower} the lower bound for range counting is $\Omega(\log n/\log \log n)$ for $O(n \polylog n)$ space, thus $t(n\log n)= \Omega(\log n/\log \log n) \Rightarrow t(n) = \Omega(\log n /\log \log n)$.
\end{proof}

Note that natural attempt is to replace the 1D range tree that we used above by a 2D range tree and perform a similar sweep procedure, but this does not seem to work for building higher dimension answer strings.

\section{Optimality}

In this section, we show that the upper bound in Bille et al.~\cite{bille-2010} is nearly optimal, for two reasons. First, it is clear that by Theorem~\ref{thm:n2lower}, the upper bound in Lemma~\ref{lem:up1} is optimal, when the space used is $O(n\polylog n)$.

Second, in the cell-probe model with words of size $\log \strlen$ we also have the following lemma by Bille et al.~\cite{bille-2010}.

\begin{lemma} \label{lem:up1}
There is a data structure for the grammar random access problem with $O(n)$ space and $O(\log \strlen)$ time. This data structure works in the word RAM with words of size $\log \strlen$.
\end{lemma}

\begin{lemma} \label{lem:up2}
There is a data structure for the grammar random access problem with $O(n)$ space and $O(n \log n / \log \strlen)$ time.
\end{lemma}
\begin{proof}
This is a trivial bound. The number of bits to encode the grammar is $O(n \log n)$ since each rule needs $O(\log n)$ bits. The cell size is $O(\log \strlen)$, so in $O(n\log n/\log \strlen)$ time the querier can just read all of the grammar. Since computation is free in the cell-probe model, the querier can get the answer immediately.
\end{proof}

Thus, we have the following corollary, by using Lemma~\ref{lem:up1} when $n=\Omega(\log^2 \strlen/\log\log \strlen)$ and Lemma~\ref{lem:up2} when the case $n=O(\log^2 \strlen/\log\log \strlen)$. This corollary implies that our lower bound of $\Omega(n^{1/2-\eps})$ is nearly the best one can hope in the cell-probe model.
\begin{corollary}
Assuming $w = \log \strlen$, there is a data structure in the cell-probe model with space $O(n)$ and time $O(\sqrt{n\log n})$.
\end{corollary}
\section{Extensions and Variants}
In this section we discuss a bit about what the lowed bound means for LZ-based compression, and ways to extend the lower bound to BWT-based compressions.

\subsection{LZ-based Compression}

The typical cases for grammar-based compression is Lempel-Ziv~\cite{ziv1977universal,lempel1976complexity}. First let us look at LZ77. For LZ77 we have the following lemma.

\begin{lemma}[Lemma~9 of~\cite{charikar2005smallest}]
The length of the LZ77 sequence for a string is a lower bound on the size of the smallest grammar for that string.
\end{lemma}

The basic idea of this lemma is to show that each rule in the grammar only contribute one entry for LZ77. Since LZ77 could compress any string with small grammar size into a smaller size, it can also compress the string $s_{\setb}$ in Lemma~\ref{lem:reduction_from_sd} into a smaller size. Thus the lower bound for grammar random access problem also holds for LZ77.

The reader might also be curious about what will happen for the LZ78~\cite{ziv1978compression} case. Unfortunately the lowerbound does not hold for LZ78. This is because LZ78 is a ``bad'' compression scheme that even the input is $0^n$ of all $0$'s, LZ78 can only compress the string to length of $\sqrt n$. But a random access on an all $0$ string is trivially constant with constant space. So we are not able to have any lower bounds for this case.

There are also lots of other variants of LZ-based compressions. As long as the compression is efficient like LZ77, we have the lower bound. Otherwise if it is like LZ78, we do not.

\subsection{Lower Bound for BWT-based Compression Access}

In last sections we talked about strings that could be compressed efficiently used by grammar-based compression scheme. But it might be an interesting question to ask if we take another compression approach, say, BWT-based compression, does the lower bound holds as well? We answer this question positively here. We claim that with a little modification used our ``hard instance'' used in last section could be efficiently compressed by BWT, so that our lower bound holds for BWT as well.

The BWT of a binary string $s \in \{0,1\}^N$ could be obtained by the following process. We use $s^k$ to denote the string which is the concatenation of the substring $s[k+1 \ldots N]+$`\$'$+s[1 \ldots k]$ where \$ is the end-of-string symbol lexicographically smaller than $0$ and $1$, and $s[1\ldots k]$ is the substring obtained by the first $k$ bits of $s$. The BWT of the string $s$ is the string formed by the last characters of list $\{ s^k \}$ for $ 0\leq k \leq \strlen$ \emph{after sorting}. This string is of length $\strlen+1$ but it will have long ``runs'', which means maximal consecutive $1$'s or $0$'s when omitting `\$'. For example in Figure~\ref{fig:bwt} there are $3$ runs `0',`111' and `00'.  We call the function defined in this process $\bwt: \{0,1\}^n \rightarrow \{0,1,\$\}^{n+1}$. And this $\bwt$ function is invertible according to~\cite{burrows1994block}, that is, given $\bwt(s)$, one can always recover $s$ in linear time.

\begin{figure}[!hbp]
\centering
\begin{minipage}[c]{2cm}
\centering
  \texttt{010110\$
  10110\$0
  0110\$01
  110\$010
  10\$0101
  0\$01011
  \$010110}
\end{minipage}
$\Rightarrow$
\begin{minipage}[c]{2.5cm}
\centering
  \texttt{\$010110 0
  0\$01011 1
  010110\$ \$
  0110\$01 1
  10\$0101 1
  10110\$0 0
  110\$010 0}
\end{minipage}
\caption{The Burrows-Wheeler Transformation for the string $S = 010110$. The strings on the right side are sorted in lexicographic order. Note that `\$' is the end-of-string character which is smaller than both $0$ and $1$. The value of the $\bwt$ function is the last column of the sorted list.} \label{fig:bwt}
\end{figure}

However, BWT itself is not a compression algorithm. But there are several approaches to compress the text efficiently after BWT, e.g., \cite{deorowicz2002second,manzini2001analysis}, first use MTF (move-to-front) encoding, and then arithmetic encoding. The compressed length is $n \log \strlen$. For binary strings, there is a much easier way to bound the number of bits for storing the compressed string. We can just save the length and an extra bit indicating $0$ or $1$ for each run, and $\log \strlen$ bits for storing the position of `\$'. Thus the compressed length is $n \log \strlen$.

\begin{lemma} \label{cor:bwtruns}
For a binary string $s$, if $\runs(s) = n$, then the BWT-based compressed representation of $s$ is less than $n \log \strlen$ bits.
\end{lemma}

In the appendix, we will prove the following lemma for a string $s'_{\setb}$ which is quite similar to the string $s_{\setb}$ used in Lemma~\ref{lem:reduction_from_blsd}.

\begin{lemma}[Sketched Version of Lemma~\ref{lem:bwtmain}]
There is a string $s'_{\setb}$ constructed for blocked set $\setb$ and a mapping $\sigma: 2^{\{0,1\}^n}$ such that for any set $\seta \in [n]$, $s'_{\setb}[\sigma(\seta)] = (\seta \cap \setb = \emptyset)$. And the length of $s'_{\setb}$ is $(4\blocksize)^\blockcount$, while $\runs(s'_{\setb})=O(n)$.
\end{lemma}

By using this lemma we prove the following main theorem for BWT, with details in appendix.

\begin{theorem} \label{thm:bwtmain}
For random access a bit in a binary string compressed by BWT-based methods with $m = n\log \strlen$ bits. Assume $w = \omega(\log \dspace)$. Let $\eps>0$ be any arbitrarily small constant. For any 2-sided-error data structure, $t \ge n/w^{\frac{1+\eps}{1-\eps}}$. And in terms of $\strlen$, $t \ge \frac{\log \strlen}{\log \dspace \cdot w^{\frac{\eps}{1-\eps}}}$.
When setting $w=\log \strlen$ and $\dspace=poly(m)$ (polynomial in $m$), there is a constant $\delta$ such that we get  $t \ge m^{1/3 - \delta}$. And in terms of $\strlen$, $t \ge (\log \strlen)^{1-\delta}$.
\end{theorem}


\section*{Acknowledgement}
We thank Travis Gagie and Pawel Gawrychowski for helpful discussions.
\bibliographystyle{plain}
\bibliography{gram}
\appendix

\section{Range Counting Lower Bound} \label{sec:rclowerbound}

In this section we prove the lower bound for range counting (Lemma~\ref{lem:2dlower}). The way of proving the lower bound for the range counting problem in \patrascu~\cite{patrascu8unifying} is by reduction from the reachability problem on butterfly graphs. However, the version used in \patrascu's paper can only be reduced to range counting on $[n] \times [n]$ grid. In order to make the reduction work for $[n] \times [n^\eps]$ grid, we need the following unbalanced butterfly graph for the reachability problem.

\begin{definition}[(Unbalanced) Butterfly Graph] \label{def:butterfly}
A graph $G=(V,E)$ is called butterfly graph iff there exists some integers $\instcount,\blockcount,\blocksize,D$ such that $|V| = \instcount \blockcount + \frac{\instcount \blockcount}{D}$, $|E| = \instcount \blockcount \blocksize$, $\log_{\blocksize} \frac{\blockcount}{D} \geq D$, and the vertices of $G$ could be labeled uniquely as $(h,b, a_1, \ldots, a_D)$ where $h \in [\instcount], b \in \{0\}\cup[D]\text{(representing the layer)}, a_1, \ldots, a_D \in [\blocksize]$. And an directed edge $(u,v) \in E$ iff $u$ labeled $(h, i, a_1, \ldots, a_{i-1}, a_i, a_{i+1}, \ldots, a_D)$ and $v$ labeled $(h, i+1, a_1, \ldots, a_{i-1}, a'_i, a_{i+1}, \ldots, a_D)$.

If we merge the vertices having difference $h$ labels in the layer $0$ of the graph together, i.e., make the vertices labeled with $(h_1, 0, a_1, \ldots, a_D)$ and $(h_2,0, a_1, \ldots, a_D)$ the same vertex, we call it an unbalanced butterfly graph. All the vertices on other layers can only have one unique label. So there are $\frac{\blockcount}{D}$ vertices in layer $0$ and $\frac{\instcount \blockcount}{D}$ vertices in other layers. For an example, the reader may refer to Figure~\ref{fig:newgraph}.
\end{definition}

Note here the $D \leq \log_{\blocksize} \frac{\blockcount}{D}$ condition is enforced to make sure that the number of different labels $\instcount (D+1) \blocksize^{D} \leq |V|=\instcount \blockcount + \frac{\instcount \blockcount}{D}$.

\begin{figure}[!hbp]
  \centering
  \includegraphics[width=7cm]{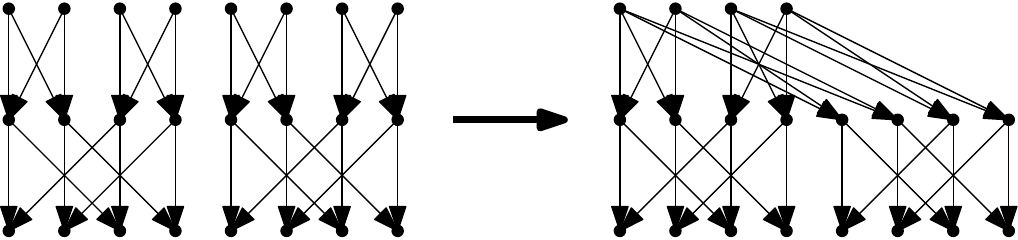}
  \caption{From butterfly graph to unbalanced butterfly graph.}
  \label{fig:newgraph}
\end{figure}

If $\instcount = 1$, the unbalanced butterfly graph is just the butterfly graph. In \patrascu's paper a relationship between reachability in the butterfly graph and range counting problem is given~\cite[Section~2.1+Appendix~A]{patrascu8unifying}. Actually the proof could be extended to the unbalanced butter fly graphs with little modifications.
The basic idea of the proof is to map an edge $(h, i, a_1, \ldots, a_{i-1}, a_i, a_{i+1}, \ldots, a_D)\rightarrow(h, i+1, a_1, \ldots, a_{i-1}, a'_i, a_{i+1}, \ldots, a_D)$ to a rectangle on the grid. This is because all the vertices in layer $0$ that leads to this edge is are the vertices $(h, 0, \star, \ldots, \star, a_i, a_{i+1}, \ldots, a_D)$, and all the vertices in the last layer this edge leads to are the vertices $(h, D, a_1, \ldots, a'_i, \star, \ldots, \star)$, where $\star$ means arbitrary values. And a reachability query from layer $0$ to the last layer could be translated to a stabbing query, which will further be translated to a range counting query.
We state the following lemma with the proof leaving to the reader.

\begin{lemma} \label{lem:rcandreachability}
A range counting data structure on a $P\times Q$ grid could be used to solve
the reachability problem on unbalanced butterfly graphs with $P$ vertices in layer $0$ and $Q$ vertices in the last layer with same space and query time. The reachability problem is to answer queries if there is a directed path from vertex $u$ to vertex $v$ where $u$ is a vertex in layer $0$ and $v$ is a vertex in the last layer.
\end{lemma}

By choosing $n=\instcount \blocksize$, $\instcount = n^{1-\eps}$ and $D = \Theta(\log n/\log \log n)$, we can see that the unbalanced butterfly graph has $\frac{\blocksize}{D} = \tilde\Theta(n^{\eps})$ vertices in layer $0$ and $\tilde\Theta(n)$ vertices in the last layer. Thus a lower bound for reachability on this graph will implies lower bound for range counting on the $[n] \times [n^\eps]$ grid.

\begin{lemma}[\cite{patrascu8unifying}]
For the reachability problem on the butterfly graphs with $\log_{\blocksize} \frac{\blockcount}{D} \geq D$, there exists some constant $c$ such that when $\blocksize \geq \max\{\left(\frac{e\dspace D}{n}\right)^c, w^c\}$, the query time $t = \Omega(D)$.
\end{lemma}

And it is easy to observe that a reachability data structure for unbalanced butterfly graphs is also a reachability data structure for butterfly graphs.

\begin{lemma}
A reachability data structure for unbalanced butterfly graphs could be used to solve reachability problem for butterfly graph using the same space and query time.
\end{lemma}
\begin{proof}
This is almost trivial to prove. For a query on vertices $(h_1, 0, a_1, \ldots, a_D)$ and $(h_2,D,a_1, \ldots, a_D)$ in the butterfly graph, it will just map to the vertices with the same label on the corresponding unbalanced butterfly graph.
\end{proof}

Finally, we have the proof for Lemma~\ref{lem:2dlower}.

\begin{proof}[Lemma~\ref{lem:2dlower}]
We know that $w = O(\log n)$ and $\dspace \leq n \log^d n$ for some constant $d$. And since $D \leq \log_{\blocksize} \frac{\blockcount}{D} \leq \log n$, we can choose $B = \log^{(d+1)c+1} n$ to make $B \geq (e\log^d n \log n)^c \geq \left(\frac{e\dspace D}{n}\right)^c$ and $B \geq \log^{c+1} n \geq w^c$. It implies that $t = \Omega(D)$. We can choose $D = \frac{1}{2(d+1)c+2} \cdot \frac{\log n}{\log \log n} \leq \log_{\blocksize} \frac{\blockcount}{D}$. Thus in the best case we have $t = \Omega(\log n/\log \log n)$.
\end{proof}

\section{Remaining Proofs for BWT-based Compressions}
Here we bound the number of ``runs'' in $s_{\setb}$ (see Lemma~\ref{lem:reduction_from_blsd} for definition). We are going to show that for the string derived by a variant of the string $s_{\setb}$, BWT-based compression schemes could compress it well. As a result, by using a similar argument of Lemma~\ref{lem:reduction_from_blsd}, randomly accessing a bit in BWT compressed strings is also hard.

An important observation for string $s_{\setb}$ is that it could also be derived by applying the following \emph{replacement rules} to the string ``1'' sequentially for $i=1,2, \ldots, \blockcount$.

\begin{itemize}
  \item $0 \rightarrow 0^{\blocksize}$;
  \item $1 \rightarrow h_i$;
\end{itemize}

where
$$ h_i = \big(\blocksize (i-1) \not\in \setb \big) \ \big(\blocksize (i-1)+1 \not\in \setb \big) \cdots \big(\blocksize(i-1)+\blocksize-1 \not\in \setb \big) $$
is a binary string. The reader could simply check that it defines the same string as $s_{\setb}$ in Lemma~\ref{lem:reduction_from_blsd}.

Let $s_0 = 1$ and $s_i$ to be the binary string obtained by applying the replacement rules for $i$ to $s_{i-1}$. Here replacement rules means that we simply replace every $0$ in $s_{i-1}$ with $0^b$ and every $1$ with $h_{i}$. And we note that $s_\blockcount = s_\setb$. And we have the following lemma for a variant of $s_\setb$.

\begin{lemma} \label{lem:bwtmain}
For every set $\setb \subseteq [\blocksize \blockcount]$, let $s_0 = 1$ and $n = \blocksize \blockcount$, for $i=1,2, \ldots, \strlen$, we apply the following replacement rule on $s_{i-1}$ to get $s_i$,
\begin{itemize}
  \item $0 \rightarrow 0^{4\blocksize}$;
  \item $1 \rightarrow h'_i$;
\end{itemize}
where the length of $h'_i$ is $4\blocksize$ and $h'_i$ is $h_i$ with $0$ replaced with $1011$ and $1$ replaced by $1101$, e.g., if $h_i = 01$ then $h'_i = 10111101$.

We let $s'_{\setb} = s'_\blockcount$, then the following properties of $s'_\setb$ is true.
\begin{enumerate}
  \item $|s'_{\setb}| = (4\blocksize)^\blockcount = \strlen$;
  \item There is a function $\sigma$ independent of $\setb$ such that for every blocked set $\seta \subseteq [\blocksize \blockcount]$, $s'_{\setb}[\sigma(\seta)] = (\seta \cap \setb = \emptyset)$;
  \item $\runs(s'_{\setb}) = O(n) = O(\blocksize \blockcount)$.
\end{enumerate}
\end{lemma}

First assuming this lemma is true, we have the proof for Theorem~\ref{thm:bwtmain}.

\begin{proof}[Theorem~\ref{thm:bwtmain}]
From Lemma~\ref{lem:bwtmain} and Lemma~\ref{cor:bwtruns} we know that one can store $s'_{\setb}$ in a string with compressed size at most $m = n\log \strlen$ bits. Pick $\delta = \frac{100\eps}{1-\eps}$, it is easy to verify the $t \geq m^{1/3-\delta}$ lower bound. All the rest are similar to the proof of Theorem~\ref{thm:grammar_lb_from_blsd_det}.
\end{proof}

Now we prove Lemma~\ref{lem:bwtmain}.

\begin{proof}[Lemma~\ref{lem:bwtmain}]
First of all, it is easy to make the following observations.
\begin{enumerate}
  \item Each $h'_i$ starts with $1$ and ends with $1$.
  \item The number of consecutive $0$'s in $s'_i$ is $(4\blocksize)^j$.
  \item The number of consecutive $1$'s in $s'_i$ is at most $4$.
  \item The length of $s'_i$ is $(4\blocksize)^i$.
  \item If $s_i^{'4r\blocksize+\alpha}$ starts with $\gamma$ $0$'s and $\gamma$ is not a multiple of $4\blocksize$ or 1, then $s_i^{'4r\blocksize+\alpha}$ ends with $0$.
  \item If we represent a blocked set $\seta$ in integer as $(a_1, a_2, \ldots, a_m)_\blocksize$ as an base $\blocksize$ integer, then we know that for the integer $\sigma(\seta) = (4a_1+2, 4a_2+2, \ldots, 4a_m+2)_{4\blocksize}$ in base $4\blocksize$, we have $s_{\setb}[\seta] = s'_{\setb}[\sigma(\seta)]$. So this is the $\sigma$ we want.
\end{enumerate}
Second, we are going to show that $\runs(s'_i) \leq \runs(s'_{i-1}) + 512\blocksize$. If this is true, then $\runs(s'_\blockcount) \leq \runs(s'_{\blockcount-1}) + 512\blocksize \leq \cdots \leq \runs(s'_0) + 512\blocksize\blockcount = 512\blocksize\blockcount$, which is the third item we need to prove.

The way of upper bounding $\runs(s'_i)$ by $\runs(s'_{i-1})$ is to see the process of computing BWT of $s'_{i}$ as first computing the BWT of $s'_{i-1}$ then inserting the rest bits. Precisely speaking, if we group the bits in $s'_i$ into segments of length $4\blocksize$ in $s'_i$ from start, we know each segment is derived by a bit in $s'_{i-1}$. A simple observation  is that by looking at the start of each segment ($s_i^{'0}, \ldots, s_i^{'(4\blocksize)^\blockcount}$), the last bit of the sorted list is the same as the BWT of $s'_{i-1}$. So there are $\runs(s'_{i-1})$ runs in the string formed by last bits of the sorted list.

So all the $\runs(s'_i) - \runs(s'_{i-1})$ runs come from the other parts of $s'_i$, say $s_i^{'4r\blocksize+\alpha}$ for $\alpha\in[4\blocksize], r\in[(4\blocksize)^{i-1}], \alpha\in [4\blocksize-1]$. We discuss about them in two cases.

\paragraph{Case I.} For $s_i^{'4r\blocksize+\alpha}$ starts with $\leq 4\blocksize$ $0$'s. We know that for any $p \in [(4\blocksize)^{i-1}]$, $s_i^{'4p\blocksize+\alpha}$ ends in the same bit as $s_i^{'4r\blocksize+\alpha}$. If we look at another string $s_i^{'4qb+\beta}$ where $\beta \neq \alpha$, then by the following lemma we know that the first $5\cdot 4\blocksize$ of them differs.

\begin{lemma} \label{lem:leq4b}
The first $5\cdot 4\blocksize$ bits in $s_i^{'4p\blocksize+\alpha}$ and $s_i^{'4q\blocksize+\beta}$ ($p, q \in [(4\blocksize)^{i-1}], \alpha, \beta \in [4\blocksize-1], \alpha \neq \beta$) are always different for different $\alpha \neq \beta$ if they do not start with more than $4\blocksize$ $0$'s.
\end{lemma}

And by the following lemma, we know the last character of $s_i^{'4r\blocksize+\alpha}$ only depends on the first $5\cdot 4\blocksize$ bits.

\begin{lemma} \label{lem:20beq}
If the first $5\cdot 4\blocksize$ characters from $s^{'4p\blocksize+\alpha}_i$ and $s^{'4q\blocksize+\beta}_i$ are the same, and they are not start with more than $4\blocksize$ 0's, then they end with the same character.
\end{lemma}

So we know that if we group these $s_i^{'4r\blocksize+\alpha}$ into groups according to the first $5 \cdot 4\blocksize$ bits, then there will be $\leq 64\cdot 4\blocksize$ groups and each group will end in the same bit. So the number of runs increased by inserting these $s_i^{'4r\blocksize+\alpha}$ is $\leq 2\cdot 64 \cdot 4\blocksize = 512\blocksize$.

\paragraph{Case II.} For $s_i^{'4r\blocksize+\alpha}$ starts with $> 4\blocksize$ $0$'s. Since $a\neq 0$, $s_i^{'4r\blocksize+\alpha}$ ends with $0$. And we know that after sorting all these $s_i^{'4r\blocksize+\alpha}$ starting with $> 4l'\blocksize$ $0$'s and $< 4(l+1)\blocksize$ $0$'s, they will be inserted into $s_i^{'4p\blocksize}$ and $s_i^{'4q\blocksize}$ for some $p$ and $q$ where one of them starts with $4l\blocksize$ $0$'s and another of them starts with $4(l+1)\blocksize$ $0$'s. By the following lemma and the fact that all the strings inserted here will be before the strings in case I, we know that at least one of them ends with $0$ so the number of runs increased is $0$.

\begin{lemma} \label{lem:notallzero}
For all possible choices of $p,q \in [(4\blocksize)^{i-1}]$, if there are $4\blocksize l$ $0$'s in the prefix of $s_i^{'4p\blocksize}$ and $4\blocksize(l+1)$ $0$'s in the prefix of $s_i^{'4q\blocksize}$ ($l \geq 1$), then at least one of them ends in $0$.
\end{lemma}

\end{proof}

At last we prove the lemmas left.

\begin{proof}[Lemma~\ref{lem:leq4b}]
We prove by contradiction. Say if there exists $p$, $q$, $\alpha$, $\beta$ such that $s_i^{'p\blocksize+\alpha}$ and $s_i^{'q\blocksize+\beta}$ have common prefix longer than $5 \cdot 4\blocksize$ and $\alpha \neq \beta$.

\begin{figure}[!hbp]
  \centering
  \includegraphics[width=10cm]{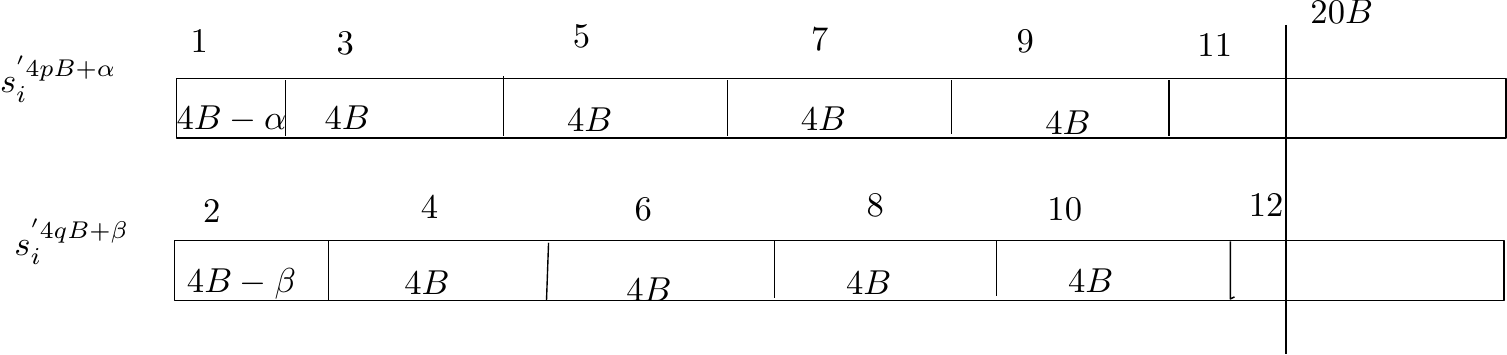}
  \caption{Alignment of $s_i^{'4p\blocksize+\alpha}$ and $s_i^{'4q\blocksize+\beta}$}\label{fig:align}
\end{figure}

According to the assumption the first $20\blocksize$ bits of them are all the same. However, we know that there are $<4\blocksize$ $0$'s at the start of the two strings. So segment 3 and segment 4 must be $h'_i$ corresponds to two $1$'s in $s'_{i-1}$. Since we know that $h'_i$ starts and ends with $1$, so segment 5 must be $h'_i$ as well. By the same argument we know that segment $6$ to segment $12$ are all $h'_i$. However there are at most 4 consecutive $1$'s in $s'_{i-1}$, so it is not possible.

For all possible choices of $p \in [(4\blocksize)^{i-1}]$ and $\alpha \in [4\blocksize]$, there are only $64\cdot 4\blocksize$ different kinds of prefixes of length $5\cdot 4\blocksize$ in $s_i^{'p\blocksize+\alpha}$. This is because these $5 \cdot 4\blocksize$ bits are derived by at most $6$ bits in $s_{i-1}$. The number of all the possible combinations of these $6$ bits is $2^6 = 64$. And $\alpha$ has $4\blocksize$ different choices, so the number of possible prefixes is $64 \cdot 4\blocksize = 256\blocksize$.
\end{proof}
\begin{proof}[Lemma \ref{lem:20beq}]
If the first $5\cdot 4\blocksize$ characters from $s^{'4p\blocksize+\alpha}_i$ and $s^{'4q\blocksize+\beta}_i$ are the same, and they are not start with more than 4b 0's.

Then by Lemma~\ref{lem:leq4b}, we know $\alpha=\beta$.

And we can also easily know that $s'_i[4p\blocksize+\alpha]$ and $s'_i[4q\blocksize+\beta]$ are generate by a same character. That is $s_{i-1}[p+1]$ and $s'_{i-1}[q+1]$, which are the (p+1)'s character and (q+1)'s character of $s'_{i-1}$, must be the same.

So the character $s'_i[4p\blocksize+\alpha]$ and $s'_i[4q\blocksize+\beta]$ are also the same. and $s'_i[4p\blocksize+\alpha]$ is the last character of $s_i^{'4p\blocksize+\alpha}$, $s'_i[4p\blocksize+\beta]$ is the last character of $s_i^{'4p\blocksize+\beta}$
\end{proof}

\begin{proof}[Lemma~\ref{lem:notallzero}]
We know that in $s'_{i-1}$, there are $l+1$ $0$'s in the prefix of  $s_{i-1}^{'p}$ and $l$ $0$'s in the prefix of $s_{i-1}^{'q}$. When $l \geq 2$, we know that they must be derived by a $0$ in $s'_{i-2}$. However, we know that $4\blocksize | l$ and $4\blocksize | l+1$ can not hold simultaneously, so at least one of them must end with $0$.

If $l=1$, we know that $s_i^{'4q\blocksize}$ starts with $8\blocksize$ $0$'s. It must be from a $0$ in $s'_{i-2}$, however, when $\blocksize > 2$ we know that $s_i^{'4q\blocksize}$ ends with $0$.
\end{proof}

\end{document}